\documentclass{article}
\usepackage{cite}
\usepackage{spconf,amsmath,graphicx}
\usepackage{verbatim}
\usepackage{graphicx}
\usepackage[utf8]{inputenc}
\usepackage{setspace}
\usepackage{euscript}
\usepackage{graphicx}
\usepackage{amsmath}
\usepackage{amsfonts}
\usepackage{amssymb}
\usepackage{commath}
\usepackage{amsthm}
\usepackage{xfrac}
\newtheorem{theorem}{Theorem}[section]
\newtheorem{corollary}{Corollary}[section]

\usepackage[colorlinks, citecolor=blue]{hyperref}

\def\x{{\mathbf x}}

\title{MULTIPLE ILLUMINATION PHASELESS SUPER-RESOLUTION (MIPS) WITH APPLICATIONS TO PHASELESS DOA ESTIMATION AND DIFFRACTION IMAGING}
\name{Fariborz Salehi, Kishore Jaganathan and Babak Hassibi}
\address{Department of Electrical Engineering, California Institute of Technology, Pasadena.}
\begin{document}
%
\maketitle
\begin{abstract}
Phaseless super-resolution is the problem of recovering an unknown signal from measurements of the ``magnitudes" of the ``low frequency" Fourier transform of the signal. This problem arises in applications where measuring the phase, and making high-frequency measurements, are either too costly or altogether infeasible. The problem is especially challenging because it combines the difficult problems of phase retrieval and classical super-resolution. Recently, the authors in \cite{K4} demonstrated that by making three phaseless low-frequency measurements, obtained by appropriately ``masking" the signal, one can uniquely and robustly identify the phase using convex programming and obtain the same super-resolution performance reported in \cite{Candes1}. However, the masks proposed in \cite{K4} are very specific and in many applications cannot be directly implemented. In this paper, we broadly extend the class of masks that can be used to recover the phase and show how their effect can be emulated in coherent diffraction imaging using multiple illuminations, as well as in direction-of-arrival (DoA) estimation using multiple sources to excite the environment. We provide numerical simulations to demonstrate the efficacy of the method and approach.
\end{abstract}
\begin{keywords}
Super-resolution, phase-retrieval, direction-of-arrival, diffraction imaging, semidefinite relaxation.
\end{keywords}
\section{Introduction}
\label{sec:intro}

It is often difficult to obtain high-frequency measurements in sensing systems due to physical limitations on the highest possible resolution a system can achieve. As an example, the fundamental resolution limit in optical systems caused by diffraction is an obstacle to observe sub-wavelength structures. {\textit{Super-resolution}} is the problem of recovering the high-frequency features of the signal using low-frequency Fourier measurements. In addition, many measurement systems can only measure the magnitude of the Fourier transform of the underlying signal. The fundamental problem of recovering a signal from the magnitude of its Fourier transform is known as {\textit{phase retrieval}}. 
Both of the aforementioned reconstruction problems have rich history and occur in many areas in engineering and applied physics such as  astronomical imaging \cite{Astro1,Astro2}, X-ray crystallography \cite{Crystal}, medical imaging \cite{Medical1,Medical2,Medical3}, and optics \cite{Optic}. A wide variety of techniques have been proposed for super-resolution \cite{SR1,SR2,Candes1,Recht} and phase retrieval \cite{Fienup,K2,Eldar1} problems.

Here we consider the {\textit{phaseless super-resolution}} problem, which is the problem of reconstructing a signal using its low-frequency Fourier magnitude measurements. Our work is inspired by \cite{K4} where it was shown that using three phaseless low frequency measurements, obtained by appropriately ``masking" the signal, one can uniquely and robustly identify the phase using convex programming and obtain the same super-resolution performance reported in \cite{Candes1}. While this is a significant result, due to physical limitations in measuring systems, it is not always possible to generate the mask matrices required in \cite{K4}. The main contribution of this paper is to broadly extend the class of masks that can be used to recover the phase using convex programming. In addition, we show how these masks can be implemented in {\textit{coherent diffraction imaging}}, using multiple illuminations, and {\textit{direction of arrival estimation}}, using multiple sources to excite the environment.

The organization of the paper is as follows. In Section 2, we mathematically set up the reconstruction problem and present our main result. In Section 3, we describe the practical significance of our result. Section 4 contains the details of the proof. The results of the various numerical simulations are presented in Section 5.

\section{Main Result}
\label{sec:problem definition}

Let $\mathrm x=(x[0],x[1],\ldots,x[N-1])$ be a complex-valued signal of length $N$ and sparsity $k$. Suppose we have a device that can only measure the magnitude-squares of the $2K+1$ low frequency terms of the $N$ point DFT of $\mathrm x$ (one DC term and $K$ lowest frequencies on either side of it). Clearly, this is not sufficient to generally recover $\mathrm x$. The idea of masked phaseless measurements is to obtain additional information by first masking the signal and then measuring the magnitude-squares of the $2K+1$ low frequency terms of its $N$ point DFT. Mathematically, masking a signal is equivalent to multiplying it by a diagonal ``mask" matrix, say $D$ \cite{Candes2, K1}.

Indeed, more than one mask is necessary if one wishes to recover general signals from such measurements. Assuming we have $R$  masks, for $0\leq r\leq R-1$, we will depict them by $D_r = \mbox{diag}(d_r[0],d_r[1],\ldots, d_r[N-1])$. The problem we are interested in is recovering $\mathrm x$ from the resulting collection of low frequency masked phaseless measurements, viz., 
\begin{equation}
\begin{aligned}
&{\text{find}}&&x\\
&{\text{subject to}}&&Z[m,r]=|\langle f_m,D_rx\rangle|^2\;\\
&\;\;\;\text{for}\;\; -K\leq m\leq K&&{\text{and}}\;\; \;0\leq r \leq R-1,
\end{aligned}
\end{equation}
where $\langle.,.\rangle$ is the standard inner product operator, $f_m$ is the conjugate of the $m$th column of the $N$ point DFT matrix and $Z[m,r]$ denotes the magnitude-square of the $m$th term of the $N$ point DFT for the $r$th mask. The index $m$ is to be understood modulo $N$, due to the nature of the $N$ point DFT.

Of course there are two issues that arise with the above problem: (1) designing a set of masks for which one can (up to a global phase) uniquely, efficiently and stably identify the signal and (2) developing an algorithm that can provably do so. Both these issues were resolved in \cite{K4} where it is shown that, under appropriate conditions,  the following three masks
\begin{equation}
D_0 = I ~~~,~~~D_1 = I+D^{(1)}~~~,~~~D_2 = I-iD^{(1)},
\end{equation}
where the diagonal entries of $D^{(1)}$ are given by
\begin{equation*}
d^{(1)}[n]=e^{i2\pi \frac{n}{N}},~~~n=0,1,\ldots,N-1,
\end{equation*}
are sufficient to uniquely identify $\mathrm X=\mathrm x\mathrm x^\star$ using the convex program 
\begin{equation}
\label{main}
\begin{aligned}
&\underset{X\in \mathbb{S}^n}{\text{minimize}}&&\|X\|_1\\
&{\text{subject to}}&&Z[m,r]={\text{trace}} (D_r^{\star}f_mf_m^{\star}D_rX)\\
&\;\;\;\text{for}\;\; -K\leq m\leq K&&{\text{and}} \;\;\;0\leq r \leq R-1\\
&&&X\succeq 0.
\end{aligned}
\end{equation}
The above convex program is obtained by the standard method of linearizing a quadratic-constrained problem by {\textit{lifting}} \cite{LIFTING1,LIFTING2, K6, CANDESPL, BALAN, SAMET, ROMBERG} the problem to the rank-one matrix $\mathrm X=\mathrm x\mathrm x^\star$ and afterwards convexifying it by relaxing the rank one constraint to a non-negativity constraint. Since the matrix we want to recover is sparse, the $l_1$-norm is used as the objective function.

\subsection{Contribution}
\label{subsec:contribution}

While the result of \cite{K4} is very nice, in many applications, the masking matrix $D^{(1)}$ is difficult to implement. Therefore, it is desirable to have more flexibility in the mask designs so as to permit more applications. We herein propose a set of 5 flexible masks. The building blocks of these masks are the diagonal matrices denoted by $D^{(l)}$, for 
$0\leq l\leq N-1$,  where the diagonal entries are
\begin{equation*}
d^{(l)}[n]=e^{i2\pi \frac{ln}{N}},~~~n=0,1,\ldots,N-1.
\end{equation*}
We are now in a position to state our main result.
\begin{theorem}
The convex program (\ref{main}) has a unique optimizer, namely $\mathrm X=\mathrm x \mathrm x^{\star}$, and thus $\mathrm x$ can be uniquely identified (up to a global phase), if 
\begin{enumerate}
\item $\Delta = \underset{0\leq i,j\leq k-1, i\neq j}\min(t_i-t_j)~\textrm{mod}~N\geq \frac{CN}{K}$, where $t_i$ for $0 \leq i \leq k-1$ are the positions of the non-zero entries of $\mathrm x$, and $C$ is a numerical constant.
\item $y[-K],\ldots,y[0],\ldots,y[K] \neq 0$, where $y$ is the $N$ point DFT of $\mathrm x$.
\item The following mask matrices are used:
\begin{equation}
\label{Define Masks}
\begin{aligned}
&&&D_0=D^{(0)}=I,\;D_1=I+D^{(l_1)},\;D_2=I-iD^{(l_1)}\\
&&&\;\;\;\;\;\;D_3=I+D^{(l_2)},\;\;\;\; D_4=I-iD^{(l_2)}.
\end{aligned}
\end{equation} 
\item $l_1$ and $l_2$ are integers that satisfy
\begin{equation}
gcd(l_1,l_2)=1,\;\;\; |l_1|+|l_2|\leq 2K.
\end{equation}
\end{enumerate}
\label{thm:main}
\end{theorem}
As we shall presently see, the masks used in the theorem are easy to implement in both DoA Estimation and Coherent Diffraction Imaging setups.

\section{Applications}
\label{sec:methodology}

\subsection{Phaseless Direction of Arrival Estimation}
\label{sssec:Direction of Arrival Estimation}

Consider the planar direction of arrival estimation setup described in Fig.~{\ref{fig:DoA Estimation}}. Suppose there are $M$ objects which can reflect waves, with the $m$th object, for $0 \leq m \leq M-1$, located at distance $r_m$ and angle $\theta_m$ from the origin. A transmitter positioned at location $-\frac{l\lambda}{2}$ on the x-axis, where $\lambda$ is the transmission wavelength, is used to transmit narrow-band waves with center frequency $f_c=\frac{c}{\lambda}$, and a uniform linear array (ULA) consisting of $2K+1$ receivers located along the $x$-axis at $(-\frac{K\lambda}{2},\ldots,0,\frac{\lambda}{2},\ldots,\frac{K\lambda}{2})$ is used for signal detection.  The direction of arrival estimation problem deals with estimating $\theta_m$, for $0 \leq m \leq M-1$, from the received signal. 
\begin{figure}[h]
\includegraphics[width=8.5cm]{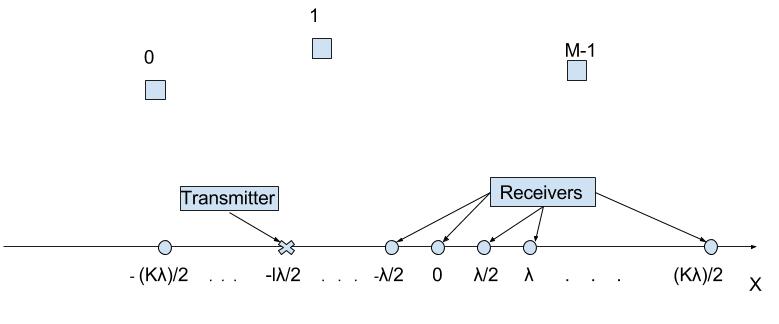}
  \caption{{Direction of arrival estimation using a uniform linear array.}}
  \label{fig:DoA Estimation}
\end{figure}


If $y$ denotes the narrow-band vector impinging on the receivers in the frequency domain, then we can write:
\begin{equation}
\begin{aligned}
&y[k]&&\propto \sum_{m=0}^{M-1}(\rho_me^{\frac{-i2\omega_cr_m}{c}})e^{i\pi (k-l) sin\theta_m},
\end{aligned}
\end {equation}
where $\rho_m$ is the reflectivity of object $m$ and $\omega_c = 2 \pi f_c$ \cite{DirectionArrival}. We refer the reader to section 6.1 of \cite{K5} to follow details of this formulation. If $l=0$, then the vector $y$ represents the $2K+1$ low-frequency terms of the Fourier series of a signal having amplitudes $\rho_me^{\frac{-i2\omega_cr_m}{c}}$ at locations $\frac{sin\theta_m}{2}$. Hence, direction of arrival estimation involves solving the classic super-resolution problem. 

Observe that, for a general $l$, the vector $y$ represents the $2K+1$ low-frequency measurements of the same signal which is masked by the matrix $D^{(l)}$. Theorem \ref{thm:main}, coupled with this critical observation, enables phaseless direction of arrival estimation: 

The mask $D_0$ in Theorem \ref{thm:main} can be implemented by putting an in-phase transmitter at the origin, $D_1$ and $D_3$ by using additional in-phase transmitters at $-\frac{l_1\lambda}{2}$ and $-\frac{l_2\lambda}{2}$, respectively, and $D_2$ and $D_4$ by using additional transmitters that have $\pi/2$ phase difference at those very locations. As a result, if $5$ strategically placed transmitters are used for transmission, then there is no need to measure phase during reception and the angles can be provably recovered by solving (\ref{main}). This is particularly useful in scenarios where measuring phase reliably is either impractical or too costly.  

{\bf Remark}: This idea also extends to the co-prime array and nested array setups described in \cite{PPV1} and \cite{PPV2}, respectively.

\subsection{Coherent Diffraction Imaging (CDI)}
\label{sssec:Diffraction Imaging}

Consider the planar CDI setup described in Fig. {\ref{fig2}}. Let the object and the detector be perpendicular to the $x$-axis, located at $x=0$ and $x=d$ respectively, and $\psi(z)$ denote the one-dimensional object which we wish to determine. The object is illuminated using a coherent source incident at an angle $\theta$ with respect to the $x$-axis. 
\begin{figure}[h]
\includegraphics[width=7cm]{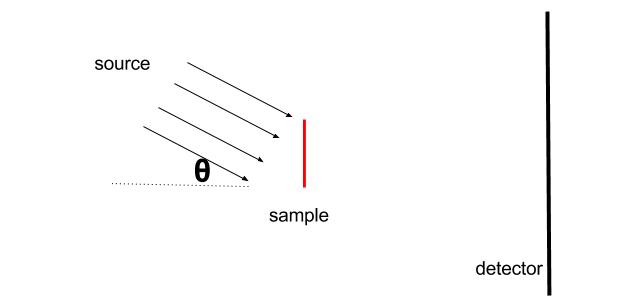}
  \caption{A typical Coherent Diffraction Imaging setup.}
  \label{fig2}
\end{figure}

Detection devices cannot measure the phase of the incoming light waves (the frequency is too high), and instead measure the photon flux. The flux measurements at position $z'$ on the detector, denoted by $I(z')$, are well approximated by:
\begin{equation}
I(z')\propto \left| \int_z \psi(z)e^{i\frac{2\pi z}{\lambda}(-\frac{z'}{d}+\theta)}dz\right|^2.
\label{eq:cdi}
\end{equation} 
If $\theta=0$, then the measurements provide the knowledge of the Fourier magnitude-square of $\psi(z)$. Section 6.2 in \cite{K5} presents details of the above formulation. Therefore, diffraction imaging involves solving the phase retrieval problem. Quite often, the approximation (\ref{eq:cdi}) only applies to positions closer to $z=0$. As a result, one needs to solve phaseless super-resolution in order to recover the underlying object. 

If $\theta=\frac{l}{d}$, then the measurements correspond to the Fourier magnitude-square of $\psi(z)$ masked by the matrix $D^{(l)}$. The equations are identical to those in the direction of arrival setup. Hence, by using $5$ strategic illuminations (using sources placed at $\theta = 0,\frac{l_1}{d},\frac{l_2}{d}$), one can provably recover the object from the low-frequency Fourier magnitude measurements by solving (\ref{main}).




\section{Proof of Theorem 2.1}
\label{sec:proof}

Let $F$ denote the $N$-point DFT matrix and $F_K$ be the $(2K+1)\times N$ submatrix of $F$, consisting of the rows $-K \leq m \leq K$ (understood modulo $N$). Also, let $y_K=F_K\mathrm x$ denote the $2K+1$ low frequency terms in the $N$-point DFT of $\mathrm x$. The proof involves two key steps: (1) the matrix $y_Ky_K^\star$ is uniquely determined by the set of constraints in (\ref{main}) and (2) given $y_Ky_K^\star$, the matrix $\mathrm x\mathrm x^\star$ can be uniquely reconstructed by minimizing $\|X\|_1$ under certain conditions.

We now provide the details for the first step. Consider the following affine transformation $Y = F_K\mathrm XF_K^\star$. When measurements are obtained using the masks proposed in Condition 3, the affine constraints of (\ref{main}) can be rewritten in terms of the variable $Y$ as: 

\begin{align}
Y[m,m] &=\abs{y_K[m]}^2,\;\;\;\;\;\; \text{for} -K\leq m\leq K \\
Y[m,m+ l_1] &=y_K[m]y_K^\star[m+ l_1], \;\text{for} -K\leq m\leq K-l_1\nonumber \\
Y[m,m+ l_2] &=y_K[m]y_K^\star[m+ l_2], \;\text{for} -K\leq m\leq K-l_2.\nonumber
\end{align}

For the sake of brevity, we omit the details here. We refer the interested readers to the proof of Theorem 3.1 in \cite{K4}. As a result, the set of constraints in (\ref{main}) can be viewed as a matrix completion problem in $Y$. Define a graph $G=(V,E)$ on the vertices $V=\{-K,-K+1,\ldots,K-1, K\}$ such that $(m,m-l_1),(m,m-l_2)\in E$ for $-K \leq m \leq K$. In other words, the graph $G$ contains an edge between vertices $i$ and $j$ if the $(i,j)$th entry of $Y$ is fixed by the measurements. Since $l_1$ and $l_2$ are co-prime (Condition 4), the graph $G$ is connected. Additionally, every vertex has an edge with itself (i.e., all the diagonal entries are fixed by the measurements). By using Corollary \ref{4.2}, we conclude that the matrix $Y=y_Ky_K^\star$ is the only feasible matrix (subject to Condition 2).

The second step is a direct consequence of the two-dimensional super-resolution theorem in \cite{Candes1} (subject to Condition 1, also known as the minimum separation condition) due to the fact that $Y$ corresponds to the $2K+1$ two-dimensional low frequencies of the two-dimensional signal $X$. 








\begin{corollary}
\label{4.2}

Suppose $G=(V,E)$ is an undirected graph on $V=\{v_0,v_1,\ldots,v_{n-1}\}$. For $e=(v_i,v_j)\in E$, define $A_e\in \mathbb{C}^{n\times n}$ as the matrix with all entries zero except for $A[i,j]$, which is equal to $1$. Also, for $i=0,1,\ldots, n-1$, define the matrix $A_i\in \mathbb{C}^{n\times n}$ as the matrix that is zero everywhere except for $A[i,i]$, which is equal to $1$. Suppose $z\in \mathbb C^n$ is a vector with non-zero entries. The matrix $Z=zz^{\star}$ is the unique solution of
\begin{equation}
\label{eq}
\begin{aligned}
&\underset{X\in \mathbb{S}^n}{\text{find}}&&X\\
&{\text{subject to}}&&{\text{trace}}(A_iX)=|z[i]|^2,\;\text{for}\;\; i=0,1,\ldots,n-1\\
&&&{\text{trace}}(A_eX)=\bar z[j]z[i],\;\text{for}\;\; e=(v_i,v_j)\in E\\
&&&X\succeq 0
\end{aligned}
\end{equation}
if and only if $G$ is connected. 
\end{corollary}
\begin{proof}
The proof of this corollary is based on the method of dual certificates. The details are omitted due to space constraints, and will be provided in the Appendix. 
\end{proof}

\section{Numerical Results}
\label{sec: numerical results}

In this section, the performance of \eqref{main} is demonstrated through numerical simulations. 

\begin{figure}[h]
\begin{center}
\includegraphics[width=8cm]{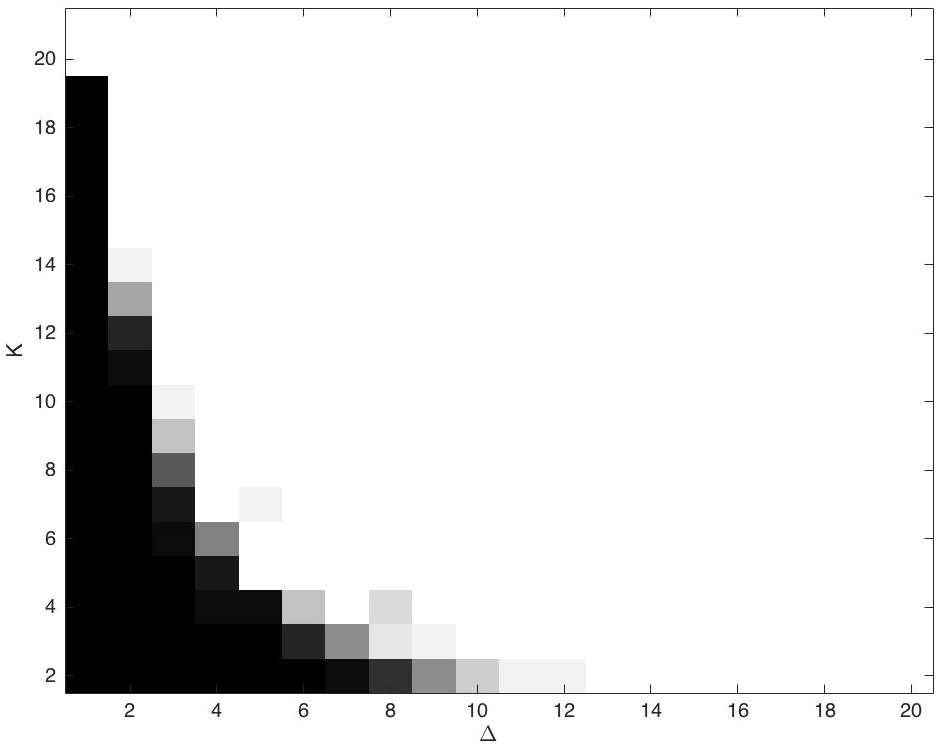}
\end{center}
  \caption{Probability of successful reconstruction for $N=20, l_1=2,l_2=3$ and various choices of $K$ and $\Delta$, using the masks defined in {\eqref{Define Masks}}.}
  \label{fig3}
\end{figure}

\subsection{Noiseless setting}

We choose $N=40$, $l_1 = 2$ and $l_2=3$. The masks $\{D_0,D_1,D_2,D_3,D_4\}$ defined in \eqref{Define Masks} are used to obtain phaseless low frequency measurements. Using parser YALMIP and solver SeDuMi, we simulate $20$ trials for various choices of $K$ and $\Delta$. We first generate the indices of the support of the signal so that the minimum separation condition is satisfied. Signal values in the support are drawn from a standard normal distribution independently. The probability of successful reconstruction of the signal by the semidefinite program (\ref{main}) as a function of $K$ and $\Delta$ is depicted in Fig. \ref{fig3}. The white region corresponds to a success probability of $1$ and the black region corresponds to a success probability of $0$. The plot shows that \eqref{main} successfully reconstructs  signals with high probability when $K\geq \frac{N}{\Delta}$. 

\begin{figure}[h]
\includegraphics[width=8cm]{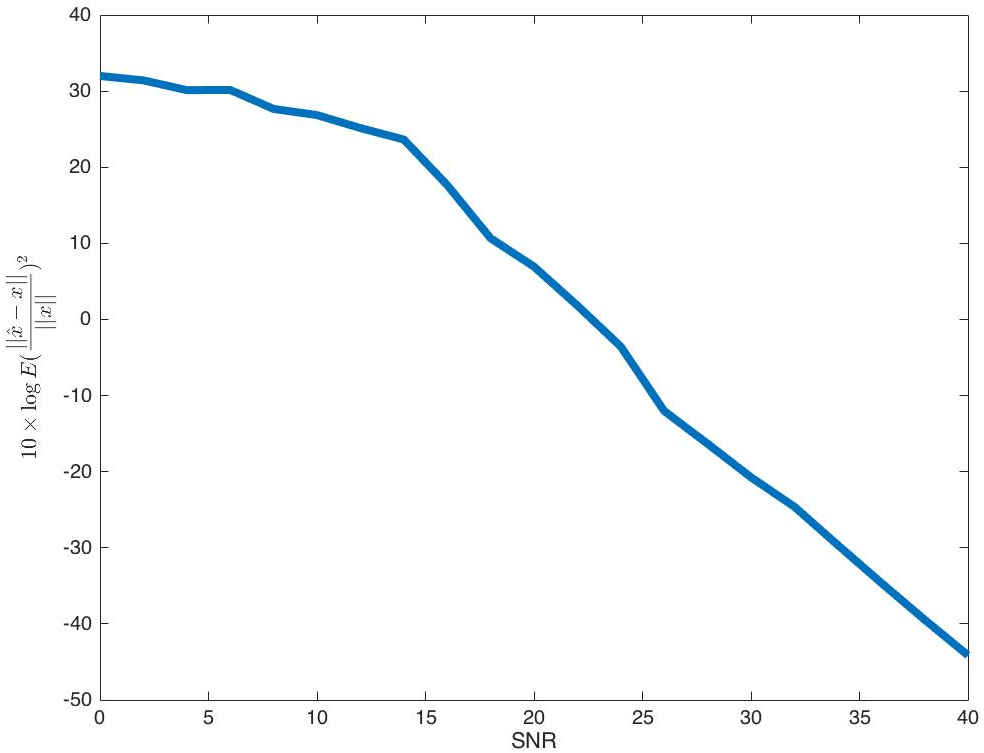}
  \caption{Mean-squared error (MSE) as a function of SNR for $N=40$, $l_1=2$, $l_2=3$, $K=14$ and $\Delta=8$.}
  \label{fig4}
\end{figure}

\subsection{Noisy setting}

A major advantage of semidefinite programming-based reconstruction is robustness to noise. In this part, we demonstrate the performance of (\ref{main}) in the noisy setting. 

For each $Z[m,r]$, we add an i.i.d. standard normal noise with appropriate variance. We first solve the program \eqref{main} by replacing the equality constraints with appropriate inequality constraints, and obtain the optimizer $\hat X$. Then, we find its best rank-one approximation, say $\hat\x\hat\x^\star$. The estimate $\hat\x$ is then compared with the true solution $\x$.

We set $N=40$, $l_1=2$, $l_2=3$, $K=14$ and $\Delta=8$. By varying the SNR, we simulate $20$ trials and compute the mean-square error $\mathbb E [\frac{\|\hat x-\mathrm x\|_2^2}{\|\mathrm x\|_2^2}]$. The results are depicted in Fig. \ref{fig4}. 

In the logarithmic scale, we see a linear relationship between the mean-squared error and SNR. This clearly shows that the reconstruction is stable in the noisy setting.


\bibliographystyle{IEEEbib}
\bibliography{refs,strings}

\section{Proof of Corollary \ref{4.2}}
\begin{proof}

The proof is based on the method of dual certificates. Let's define matrix $W\in \mathbb{C}^{n\times n}$ as follows:

$$W=\sum_{i,j: (v_i,v_j)\in E}W_{ij}$$

We define $W_{ij}$ for $0\leq i,j\leq n-1$ as follows:

$$W_{ij}=w_{ij}w_{ij}^{\star}, w_{ij}=\bar z[j]e_i-\bar z[i] e_j,$$

Where $e_i\in \mathbb C^n$ is a standard basis vector that has $1$ in $i-{th}$ entry and $0$ everywhere else. We will show that $W$ has the following properties:

\begin{enumerate}
\item $W \succeq 0,$
\item ${\text{trace}}(WZ)=0,$
\item ${\text{rank}}(W)= n-1.$
\end{enumerate}

$W$ is a positive semidefinite matrix because it is the sum of $W_{ij}$ and  $W_{ij}=w_{ij}w_{ij}^{\star}\succeq 0$. In order to show properties 2 and 3, we show the following:

$$y^{\star}Wy=0 \Leftrightarrow y= \alpha z \;\;\text{for some} \;\;\alpha \in \mathbb C.$$

One can write:

$$y^{\star}Wy=\sum_{i,j: (v_i,v_j)\in E}y^{\star}W_{ij}y=\sum_{i,j: (v_i,v_j)\in E}|y[i]z[j]-y[j]z[i]|^2$$

Therefore, 

\begin{equation}
y^{\star}Wy=0 \Leftrightarrow y[i]z[j]-y[j]z[i]=0,\;\;\; \forall (i,j)\in E
\label{10}
\end{equation}

If G is connected and the entries of $z$ are non-zero, \eqref{10} is valid {\textit{if and only if}} $y= \alpha z$ for some $\alpha \in \mathbb C$. This shows that rank($W$)\;=$\;n-1$. Also,

$$\text{trace}(WZ)=\text{trace}(Wzz^{\star})=z^{\star}Wz=0.$$

Next, let's use the above properties to prove Corollary \ref{4.2}. We want to show that the matrix $Z$ is the unique solution of  

\begin{equation}
\label{eq}
\begin{aligned}
&\underset{X\in \mathbb{S}^n}{\text{find}}&&X\\
&{\text{subject to}}&&{\text{trace}}(A_iX)=|z[i]|^2,\;\text{for}\;\; i=0,1,\ldots,n-1\\
&&&{\text{trace}}(A_eX)=\bar z[j]z[i],\;\text{for}\;\; e=(v_i,v_j)\in E\\
&&&X\succeq 0
\end{aligned}
\end{equation}

The dual of this optimization problem is 

\begin{equation}
\begin{aligned}
&\underset{\lambda \in \mathbb C^n,\mu \in  \mathbb C^{|E|}}{\text{maximize}}&&-\sum_{i=0}^{n-1}\lambda_i|z[i]|^2-\sum_{i,j:(v_i,v_j)\in E} (\mu_{i,j}\bar z[j]z[i]+\bar \mu_{i,j}\bar z[i]z[j])\\
&{\text{subject to}}&&\sum_{i=0}^{n-1}\lambda_i A_i+\sum_{i,j:e=(v_i,v_j)\in E} (\mu_{i,j}A_e+\bar \mu_{i,j} {A_e}^{\star})\succeq 0\\
\end{aligned}
\end{equation}

For $0\leq i \leq n-1$ define $N(i)$ as the set of neighbors of node $v_i$ in $G$. If we choose $\lambda_i^{*}= \sum_{j: j\in N(i)}|z[j]|^2$ and $\mu_{ij}^{*}=\bar z[j]z[i]$, then we have:

$$W=\sum_{i=0}^{n-1}\lambda_i^{*} A_i+\sum_{i,j:e=(v_i,v_j)\in E} (\mu_{i,j}^{*}A_e+\bar \mu_{i,j}^{*} {A_e}^{\star})$$

Property 1 of the matrix $W$, ensures that $W\succeq 0$ which is the dual feasibility. Property 2 is the complimentary slackness. These two properties prove that $Z=zz^{\star}$ is an optimal solution for \eqref{eq}. 

Now suppose there is another solution, namely $Z+H$, Where $H\in \mathbb S^{n}$ is an $n\times n$ Hermitian matrix. Let $T_z$ denote the set of Hermitian matrices of the form

$$T_z=\{zh^{\star}+hz^{\star}: h\in \mathbb C^n\},$$

and $T_z^{\perp}$ be its orthogonal complement. In other words, $T_z$ is the tangent space at $zz^{\star}$ to the manifold of Hermitian matrices of rank one. $H$ can be decomposed as two parts $H_{T_z}$ and $H_{T_z^{\perp}}$, which are the projections of $H$ onto the subspaces $T_z$ and $T_z^{\perp}$, respectively. In order to be an optimal solution $H$ should satisfy

\begin{equation}
{\text{trace}}(WH)={\text{trace}}(WH_{T_z})+{\text{trace}}(WH_{T_z^{\perp}})=0.
\end{equation}

Property 2 ensures that ${\text{trace}}(WH_{T_z})=0$, therefore ${\text{trace}}(WH_{T_z^{\perp}})=0$. Since $H$ is positive semidefinite its projection onto $T_z^{\perp}$ is also positive semidefinite. $H_{T_z^{\perp}}\succeq 0$ together with properties 2 and 3 lead to

$${\text{trace}}(WH_{T_z^{\perp}})=0\Leftrightarrow H_{T_z^{\perp}}=0.$$

Therefore, it remains to show that $H_{T_z}=0$. In order to be a feasible point, $H_{T_z}$ must satisfy the following conditions:

\begin{equation}
\label{eqH}
\begin{aligned}
&&&{\text{trace}}(A_iH_{T_z})=0,\;\text{for}\;\; i=0,1,\ldots,n-1\\
&&&{\text{trace}}(A_eH_{T_z})=0,\;\text{for}\;\; e=(v_i,v_j)\in E.
\end{aligned}
\end{equation}

It is easy to check that the only matrix in $T_z$ which satisfies the above conditions is 0. Therefore, $H=0$ and $Z=zz^{\star}$ is the unique solution of \eqref{eq}. 

\end{proof}

\end{document}